\documentclass[11pt]{article}
 \usepackage{amsmath}
 \usepackage{amsfonts}
 \usepackage{amssymb}
 \usepackage{amsthm}
 
\usepackage{a4wide}

\theoremstyle{plain}
\newtheorem{theorem}{Theorem}
\newtheorem{lemma}{Lemma}

\theoremstyle{definition}

\theoremstyle{remark}

\usepackage{mathrsfs}
\usepackage{multicol}
\usepackage{authblk}

\newcommand{\LD}{\backslash}
\newcommand{\Neg}{\mathord{\sim}}
\newcommand{\LDI}{\LD^{\! *}}
\newcommand{\To}{\mathbin{\Rightarrow}}
\let\phi\varphi

\begin{document}
\title{First Degree Entailment with Group Attitudes and Information Updates\thanks{This is a preprint of an article to appear in the proceedings of the 7th International Conference on Logic, Rationality and Interaction (LORI-VII), Chongquing 2019, to be published by Springer. This work was supported by the Czech Science Foundation grant GJ18-19162Y for the project \textit{Non-classical logical models of information dynamics}. The authors are grateful to three anonymous referees for their feedback.}}
\author[1]{Igor Sedl\'ar}
\author[1,2]{V\'it Pun\v{c}och\'a\v{r}}
\author[1]{Andrew Tedder}
\affil[1]{The Czech Academy of Sciences, Institute of Computer Science\\ Pod Vod\'arenskou v\v{e}\v{z}\'i 271/2, Prague, The Czech Republic}
\affil[2]{The Czech Academy of Sciences, Institute of Philosophy\\
Jilsk\'a 352/1, Prague, The Czech Republic}


\maketitle              
\begin{abstract}
We extend the epistemic logic with De Morgan negation by Fagin et al. (Artif. Intell. 79, 203--240, 1995) by adding operators for universal and common knowledge in a group of agents, and with a formalization of information update using a generalized version of the left division connective of the non-associative Lambek calculus. We provide sound and complete axiomatizations of the basic logic with the group operators and the basic logic with group operators and updates. Both logics are shown to be decidable. 

\end{abstract}

\section{Introduction}
Belnap's epistemic interpretation of First Degree Entailment \cite{Belnap1977a,Belnap1977b} shows that FDE is useful for reasoning about incomplete and potentially inconsistent information. FDE is not, however, an epistemic logic in the standard sense since its language does not contain operators expressing epistemic attitudes of agents. Such an extension of FDE was provided by Levesque \cite{Levesque1984} and brought closer to classical epistemic logic by Fagin et al.\ \cite{Fagin1995a}. These frameworks were originally put forward as an attempt to avoid the logical omniscience problem of classical epistemic logic, and so, to keep unnecessary complications out of the picture, they do not contain any additional operators utilized in the successful applications of classical epistemic logic, such as group epistemic operators \cite{Fagin1995} or operators expressing  various kinds of information update \cite{vanDitmarsch2008,Benthem2011}. 

In this paper we extend the framework of Fagin et al. \cite{Fagin1995a} with operators expressing universal and common knowledge in a group of agents (Section \ref{sec: groups}) and with a conditional operator, coming from the Non-associative Lambek Calculus, expressing information update (Section \ref{sec: updates}). These two basic logics are axiomatized and shown to be decidable; extensions are briefly mentioned, but are mostly left for future work (which is discussed in Section \ref{sec: conclusion}.)

\textit{Related work.} Non-classical modal logic with epistemic and information-dynamic operators is underdeveloped. Girard and Tanaka \cite{Girard2016} study a paraconsistent logic containing explicit revision operators. (This paper follows up on \cite{Restall1995} and \cite{Mares2002}, but these do not discuss Hintikka-style epistemic logics with revision operators; rather, they consider paraconsistent versions of AGM-style belief revision.) Rivieccio \cite{Rivieccio2014} studies an FDE-based version of Public Announcement Logic. Both of these papers contain only single-agent epistemic operators. An FDE-based group epistemic logic with universal and common knowledge is a fragment of paraconistent Propositional Dynamic Logic studied in \cite{Sedlar2016,Sedlar2019a}. B\'ilkov\'a et al.\ \cite{Bilkova2016} outline an extension of their substructural epistemic framework with common knowledge, but completeness is left for future research. The relation between substructural logic and classical information dynamics is studied in \cite{Benthem1991,Benthem2008} and \cite{Aucher2016}, for example; \cite{Dunn2001a,Restall1995a} discuss an information-dynamic interpretation of the Routley--Meyer semantics for some substructural logics. Restall \cite{Restall1995a} considers a ternary relation between sets of situations, but the framework considered in Section \ref{sec: updates} is original to this paper.

\section{FDE with group epistemic operators}\label{sec: groups}
In this section, we add to the framework of FDE with material implication, based on \cite{Fagin1995a}, modal operators representing universal knowledge in groups of agents (``everyone knows that ...'') and common knowledge. Firstly, we provide the basic definitions (Subsection \ref{sec: groups-definitions}), then we discuss the informal interpretation of the framework (\ref{sec: groups-informal}) and our  technical results, namely, a weakly complete axiomatization and a decidability result for the basic logic of the framework (\ref{sec: groups-results}). The proof is given in the technical appendix. 

\subsection{Group language and group frames}\label{sec: groups-definitions}
Fix a finite non-empty set $Ag$ (``agents'') and a countable set $Prop$ of propositional variables. The language $\mathcal{L}_{Gr}$ of FDE with material implication and group modalities contains
\begin{itemize}
\item unary connective $\Neg$ (De Morgan negation)
\item binary connectives $\land, \lor$ and $\supset$ (lattice conjunction and disjunction, material implication);
\item unary operators $K_G$, $K^{*}_G$ for each non-empty $G \subseteq Ag$ (group epistemic modalities)
\end{itemize}
Fix any $p \in Prop$ and define $\top := p \supset p$, $\bot := \Neg\top$ and $\neg\phi := \phi \supset \bot$. Formulas $\Neg\phi$ are read ``$\phi$ is false'' and $\neg\phi$ as ``$\phi$ is not true''; in our setting, these will not be equivalent. Sets $G \subseteq Ag$ represent groups of agents; $K_G\phi$ is read ``Every agent in $G$ knows that $\phi$'' and $K^{*}_G\phi$ as ``It is common knowledge in $G$ that $\phi$''. We define $K_a\phi := K_{\{ a \}}\phi$ and read this as ``Agent $a$ knows that $\phi$''.

\emph{Group frames} are $\langle S, \{ R_a \}_{a \in Ag}, \star \rangle$ where each $R_a$ is a reflexive binary relation on $S$ and $\star$ is a unary function of period two (that is, $\star(\star(x)) = x$ for all $x \in S$). We usually write $x^\star$ instead of $\star(x)$. Moreover, we define
\begin{align*}
R_G & := \bigcup_{a \in G} R_a &
R_G^{*} := \big( R_G \big)^{*}
\end{align*} 
Group models add to group frames a valuation function $v : Prop \to \mathscr{P}(S)$. For each model with $v$, we define the satisfaction relation $\vDash_v$ as usual when it comes to propositional variables and Boolean connectives; moreover, we require that
\begin{align*}
x &\vDash_v \Neg\phi \text{ iff } x^{\star} \not\vDash_{v} \phi\\
x & \vDash_v K_G \phi \text{ iff } \forall y (R_G xy \implies y \vDash_v \phi)\\
x & \vDash_v K^{*}_G \phi \text{ iff } \forall y (R^{*}_G xy \implies y \vDash_v \phi)
\end{align*}
We sometimes use the notation $v(\phi) = \{ x \mid x \vDash_v \phi \}$. Formula $\phi$ is valid in a model with $S$ and $v$ iff $v(\phi) = S$; it is valid in a frame iff it is valid in all models based on the frame and it is valid in a class of frames iff it is valid in all frames in the class. This notion of validity will be used throughout the paper. For any language $\mathcal{L}$, the $\mathcal{L}$-theory of a class of frames is the set of all $\mathcal{L}$-formulas valid in the class of frames.

It is easily seen that $v(\top) = S$ and so $v(\neg \phi) = S \setminus v(\phi)$. Hence, even though Boolean negation is not a primitive connective of our language, it can be expressed using material implication and De Morgan negation.

\subsection{Informal interpretation}\label{sec: groups-informal}
In group frames, elements of $S$ are called \textit{situations} and can be seen as situations in the sense of Barwise and Perry \cite{Barwise1983}, either concrete ones (parts of the world) of abstract ones (representations of parts of the world, either accurate or inaccurate). Mares \cite{Mares2004} discusses situations in the presence of De Morgan negation and we follow his interpretation, according to which situations may be incomplete (some $\phi$ is neither true nor false, i.e.\ neither $\phi$ nor $\Neg\phi$ is satisfied in the situation) and inconsistent (some $\phi$ is both true and false); we note that Barwise and Perry also allow ``incoherent'' situations \cite[96]{Barwise1983}. Levesque \cite{Levesque1984} uses the concept of a situation in a similar way; we note that this interpretation of the elements of $S$ is consistent with Belnap's interpretation in terms of ``simple databases'' \cite{Belnap1977a,Belnap1977b}. Existence of incomplete and inconsistent situations follows from our truth condition for $\Neg\phi$ in terms of ``the Routley star'' $\star$, which is thought of as an operation assigning to each situation its \textit{dual}; intuitively, the dual situation of $x$ makes true everything that is not made false in $x$ and vice versa. In general, we read $x \vDash_v \phi$ as ``$\phi$ is true in situation $x$ (on $v$)'', or ``The information that $\phi$ is supported by $x$ (on $v$)''.

The informal interpretation of ``epistemic accessibility relations'' $R_a$ differs only slightly from the standard reading of Kripke models for classical epistemic logic. Our basic idea is that, for each situation $x$ and each agent $a$, there is a \emph{part of $x$ that is available to $a$} in the sense that $a$ knows that it is a part of $x$. For instance, of the situation comprising the building in which my department is situated, only the part comprising my office is available to me at the moment, but upon receiving information from a colleague about something happening on a different floor, a bigger part of the situation becomes available to me. The fact that $R_axy$ is taken to mean, informally, that the part of $x$ available to $a$ is included in $y$. Hence, our truth condition for $K_a \phi$ means that $K_a \phi$ is supported (true) in $x$ iff each situation that contains the part of $x$ available to $a$ supports $\phi$---we may say that $K_a \phi$ is supported in $x$ iff the information available to $a$ in $x$ supports $\phi$. 

A note of caution is in order here, however. The elements of our models correspond to \emph{prime} situations in the sense that $x$ supports a disjunction iff it supports one of its disjuncts. ``Parts'' of situations, as we use the term, may not be prime in this sense. For instance, each prime situation containing the fact that Ann has one sibling contains the fact that Ann has one brother or the fact that Ann has one sister, but only the information that Ann has one sibling may be available to me, without me knowing if the sibling is male or female. A disjunction may be supported by a part of a situation without either disjunct being supported by \textit{that} part. ``Parts'' of situations in this sense are not necessarily elements of the model, but they may be represented by \emph{sets} of elements of the model; intuitively, the set representing a particular ``partial'' situation comprises all prime situations in the model that contain all the information in the partial situation. For instance, the partial situation supporting only the information that Ann has one sibling can be represented by the set comprising two prime situations differing in the gender of the sibling. See \cite{Belnap1977a,Belnap1977b} for details. Hence, we may speak of $R_a(x) := \{ y \mid R_axy \}$ as representing the part of $x$ available to $a$---it follows from reflexivity of $R_a$ that each $\phi$ supported by all situations in $R_a(x)$ is supported by $x$.

Let us turn now to the relations used in the satisfaction clauses for group operators. The fact that $R_G xy$ means that $y$ contains the part of $x$ available to some $a \in G$. Hence, $K_G \phi$ is supported in $x$ iff \textit{all} agents in $G$ have information that supports $\phi$. The fact that $R_G^{*} xy$ means that $(x,y)$ is in the reflexive transitive closure of $R_G$. (In fact, speaking of transitive closure is sufficient as all the relations are reflexive; we speak of reflexive transitive closure out of custom). In other words, there is a finite path $z_0 = x, z_1, \ldots, z_{n-1}, z_n = y$ such that, for all $k \in \{ 0, \ldots, n-1 \}$, $(z_k, z_{k+1}) \in R_a$ for some $a \in G$. Note that $(x,z) \in R_a$ and $(z,y) \in R_b$ means that $z$ contains the $a$-part of $x$ and $y$ contains the $b$-part of $y$. This means that $K_aK_b \phi$ is supported in $x$ iff the $a$-part of $x$ ``says'' that the $b$-part of $x$ supports $\phi$. In other words, $a$ knows that $b$ knows that $\phi$. Hence, $K^{*}_G \phi$ is supported in $x$ iff, in a standard manner, each agent in $G$ knows that all the agents know that ... all the agents know that $\phi$.

Belnap \cite{Belnap1977a,Belnap1977b} motivated FDE as a logic useful for reasoning about simple databases containing potentially inconsistent information; this reasoning involved only information formulated using $\Neg, \land$ and $\lor$. The epistemic extension of FDE by Fagin et al.\ \cite{Fagin1995a} can be seen as a logic for reasoning about potentially inconsistent databases where the relevant information may involve $K_a$, that is, where \textit{information about information} available to individual agents is involved. Here inconsistency may be encountered at least on two levels. Firstly, a database may contain inconsistent information about the information of agent $a$, that is, it may contain $K_a \phi$ and $\Neg K_a \phi$ for some $\phi$. In contrast to epistemic logic based on classical logic, the framework of \cite{Fagin1995a} allows to reason with such databases without ``explosion'', i.e.\ without inferring any $\psi$ whatsoever. Secondly, a database may contain information that the information of agent $a$ is inconsistent, that is, it may contain $K_a \phi$ and $K_a \Neg \phi$ for some $\phi$. In contrast to classical epistemic logic, the framework of \cite{Fagin1995a} does not force the conclusion that, in this case, $K_a \psi$ holds for any $\psi$ whatsoever. The upshot of our group FDE is that these features are lifted to group epistemic notions---we have here a logic useful for reasoning about potentially inconsistent information, including information about information available to groups of agents that may turn out to be inconsistent on the two levels mentioned above in connection to individual knowledge operators.

\subsection{Completeness and decidability}\label{sec: groups-results}

The axiom system $GrFDE$ contains the following axiom schemata and rules ($X \in \{ K, K^{*} \}$): 
\begin{multicols}{2}
\begin{itemize}\raggedcolumns
\item[] (A0) Any fixed axiomatization of the $\{ \land, \lor, \supset\}$-fragment of classical propositional logic
\item[] (A1) $\phi \supset \Neg\Neg \phi$
\item[] (A2) $\Neg\Neg\phi \supset \phi$
\item[] (A3) $(\Neg \phi \land \Neg \psi) \supset \Neg (\phi \lor \psi)$
\item[] (A4) $\Neg (\phi \land \psi) \supset (\Neg \phi \lor \Neg \psi)$
\item[] (A5) $X_G \phi \land X_G \psi \supset X_G (\phi \land \psi)$
\item[] (A6) $X_G \phi \supset \phi$
\item[] (A7) $K_G \phi \supset\!\subset \bigwedge_{a \in G} K_a \phi$
\item[] (A8) $K^{*}_G \phi \supset K_G (\phi \land K^{*}_G \phi)$
\item[] (R0) Modus Ponens
\item[] (R1) $\dfrac{\phi \supset \psi}{\Neg \psi \supset \Neg \phi}$
\item[] (R2) $\dfrac{\phi}{X_G \phi}$
\item[] (R3) $\dfrac{\phi \supset K_G (\psi \land \phi)}{\phi \supset K^{*}_G \psi}$
\end{itemize}
\end{multicols}

\begin{theorem}\label{thm: GrFDE complete decidable}
$GrFDE$ is a sound and weakly complete axiomatization of the $\mathcal{L}_{Gr}$-theory of all group frames. The theory is decidable.
\end{theorem}

\noindent Since Boolean negation is expressible in our language, Theorem \ref{thm: GrFDE complete decidable} can be established using the standard technique (\cite[Ch.\ 3.1]{Fagin1995}). In the technical appendix, we give an alternative ``modular'' proof, based on \cite{Nishimura1982}, that does not invoke Boolean negation and, as such, can be used in a setting where Boolean negation is not expressible (e.g.\ when specific weaker negations are used instead of De Morgan negation; see Sect.\ \ref{sec: conclusion}). 

\section{Almost arbitrary information updates}\label{sec: updates}
In this section, we extend our framework with a formalization of information update. Instead of focusing on one specific notion of update, such as public announcements, belief revision or the various notions of belief upgrade, we provide a somewhat more general account. Taking inspiration from van Benthem \cite{Benthem2014}, we add to our semantics an \emph{abstract representation of updates} and we study the general framework arising from this addition. (See also \cite{Holliday2012} for a nicely generalizable framework, based on abstract update relations, for the fragment of Public Announcement Logic closed under substitution; both frameworks bear some similarity to the general semantics for conditional logics \cite{Chellas1975}.) An interesting endeavour is to relate the abstract semantics to known notions of update via special cases of the general framework, but we leave such investigations for future work.

Similarly to the framework of \cite{Benthem2014}, information updates are represented as binary relations between elements of the model indexed by subsets of the model. Instead of pointed models in van Benthem's ``update universe'', elements of our models are prime situations. This feature of the model derives from the goal of formulating a general representation of information update on an inconsistency-tolerant background. The indexing set of situations, ``the proposition triggering the update'' \cite[32]{Benthem2014}, corresponds to the information content of the update. We do not assume the content of an update to correspond to a prime situation; typically the ``incoming'' information corresponds to a part of a prime situation. (Recall that parts of prime situations are represented in our framework by sets of prime situations.)

Hence, an \emph{update relation} on a set of situations $S$ is a function from the power set of $S$ (all possible ``triggering propositions'') to binary relations on $S$ (``situation transitions''). Equivalently, we may represent an update relation by $R \subseteq (S \times \mathscr{P}(S) \times S)$ ($RxYz$ iff $(x,y)$ is in the transition determined by the triggering proposition $Y$). In what follows, \emph{group update frames} are $\langle S, \{ R_a \}_{a \in Ag}, R, \star\rangle$ where $R$ is such an update relation. 

In modal logics of information update we typically have formulas specifying the results of information update depending on the nature of the ``triggering proposition''; in general, the interesting feature is whether updates of a certain kind are guaranteed to lead to outputs satisfying specific formulas. Here we will distinguish updates with based on information supported by the ``triggering proposition''.  

The language $\mathcal{L}_{GrUp}$ extends $\mathcal{L}_{Gr}$ with a binary connective $\LD$; formulas $\phi\LD\psi$ are read ``After updating with any  information supporting $\phi$, $\psi$ will hold''. \emph{Group update models} add a valuation function $v$ to group update frames and the satisfaction relation $\vDash_v$ is defined as usual; for $X \subseteq S$, $X \vDash_v \phi$ means that $x \vDash_v \phi$ for all $x \in X$. The new clause in the definition of $\vDash_v$ is the following:
\[
x \vDash_v \phi \LD \psi \text{ iff } (\forall Y)(\forall z)\big( (RxYz \And Y \vDash_v \phi ) \To z \vDash_v \psi \big) 
\] 
Validity is defined as before. Note that $\LD$ is a generalized version of the left division operator of the Non-Associative Lambek Calculus \cite{Dosen1992,Kurtonina1994,Restall2000}. There the truth condition uses individual situations $y$, not sets of situations.

We read $RxYz$ as ``Updating $x$ with the partial situation $Y$ may result in $z$''. Hence, $\phi\LD\psi$ is true in $x$ iff $\psi$ holds in every possible result of updating $x$ by a partial situation that supports $\phi$.

The proof system $GrUpFDE$ extends $GrFDE$ with
\begin{itemize}
\item[] (A9) $(\chi \LD \phi \land \chi \LD \psi) \supset \chi \LD (\phi \land \psi)$
\item[] (R4) $\dfrac{\phi_1 \supset \psi_1 \quad \phi_2 \supset \psi_2}{\psi_1 \LD \phi_2 \supset \phi_1 \LD \psi_2}$
\item[] (R5) $\dfrac{\phi}{\psi \LD \phi}$
\end{itemize}

\begin{theorem}\label{thm: GrUpFDE complete decidable}
$GrUpFDE$ is a sound and weakly complete axiomatization of the $\mathcal{L}_{GrUp}$-theory of all group update frames. The theory is decidable.  
\end{theorem}

Using Boolean negation, we may define a ``diamond version'' of the update operator $\LD$ as $\phi \circ \psi := \neg (\phi \LD \neg \psi)$. It is clear that 
\[
x \vDash_v \phi \circ \psi \text{ iff } (\exists Y)(\exists z)\big( RxYz \And Y \vDash_v \phi \And z \vDash_v \psi \big)
\]
Note that the connective $\circ$ is not what is usually called \emph{fusion} in the literature on substructural logic; the update operator $\LD$ is not a residual of $\circ$. An axiomatization of the theory of all group update frames in languages where $\circ$ is present as a primitive operator and Boolean negation is not expressible is an open problem. (This is the case even for the language $\{ \land, \lor, \LD, \circ \}$ and the $\langle S, R\rangle$-reducts of group update frames.)

\section{Conclusion}\label{sec: conclusion}

In this paper we outlined two FDE-based epistemic logics, the basic logic with universal and common knowledge, and its extension with a generalized left division operator of the Non-associative Lambek Calculus, formalizing an abstract notion of information update. We established axiomatization and decidability results for these logics.

Among topics that we leave out of the present paper is a study of axiomatic extensions of $GrFDE$ and $GrUpFDE$. It is especially natural to consider extensions of $GrFDE$ by various introspection axioms, such as positive introspection $K_a \phi \supset K_aK_a\phi$, Boolean negative introspection $\neg K_a\phi \supset K_a \neg K_a \phi$ and De Morgan negative introspection $\Neg K_a \phi \supset K_a \Neg K_a \phi$. Regarding extensions of $GrUpFDE$, it is interesting to take a look at how our framework accommodates some typical properties of special cases of information update (e.g.\ monotonicity $\phi \LD \chi \supset \phi \LD (\psi \LD \chi)$ or ``success'' $\phi \LD \psi \supset \phi \LD (\phi \land \psi)$; the latter seems to require an extension of our frames with a partial order on the set of situations in the style of the Routley--Meyer semantics for substructural logics \cite{Restall2000}.)

Another topic for future research are specific language extensions of our logics. A particular instance is related to the \emph{iterated update} operator $\LDI$, where $\phi \LDI \psi$ is read as ``$\psi$ holds after any finite number of updates by $\phi$''. A natural semantics for this operator is obtained by defining 
\[
R^{1} xYz := RxYz
\qquad
R^{n+1} xYz := \exists Uv (RxUv \And R^{n} vYz)
\]
and
\[
R^{*} := \{ \langle x, Y, z\rangle \mid (\exists n \in \mathbb{N})(R^{n}xYz) \}
\]  
and requiring that 
\[
x \vDash_{v} \phi \LDI\psi \text{ iff } \forall Yz ((R^{*}xYz \And Y \vDash_{v} \phi) \implies z \vDash_{v} \psi) 
\]
We conjecture that a complete axiomatization of the theory of all group update frames with $R^{*}$ is obtained by adding to $GrUpFDE$ the following: 
\begin{multicols}{2}\raggedcolumns
\begin{itemize}
\item[] (A10) $(\chi \LDI \phi \land \chi \LDI \psi) \supset \chi \LDI (\phi \land \psi)$
\item[] (A11) $\phi \LDI \psi \supset (\phi \LD \psi \land \phi \LD (\phi \LDI \psi))$
\item[] (A12) $(\phi \LD (\phi \LDI \psi)) \supset (\phi \LDI \psi)$
\item[] (R6) $\dfrac{\phi_1 \supset \psi_1 \quad \phi_2 \supset \psi_2}{\psi_1 \LDI \phi_2 \supset \phi_1 \LDI \psi_2}$
\item[] (R7) $\dfrac{\phi \supset \psi \LD \phi}{\phi \supset \psi \LDI \phi}$
\end{itemize}
\end{multicols}
(On some assumptions concerning the update relation $R$, $\phi \LDI \psi$ can be expressed in a language containing fusion and the Kleene star operator; see \cite{Bimbo2005}. Our setting intends to be more general. Also, the presence of Boolean negation, some of these assumptions concerning $R$ lead to undecidability; see \cite{Kurucz1995}. It was shown in \cite{Miller2005} that the classical Public Announcement Logic with an operator for iterated announcements is undecidable. Hence, the question is, which notions of update admit a decidable logic with iterated updates? Our general setting is especially suitable for such investigations, but they need to be left for future research.)

Another interesting topic are generalizations of the framework using weaker notions of negation than De Morgan negation used here. In general, negation can be seen as a negative modal operator with the satisfaction clause
\[
x \vDash_v \Neg \phi \text{ iff } \forall y(R_{\Neg}xy \implies y \not\vDash_v \phi) 
\]
using an arbitrary binary relation $R_{\Neg}$. If this relation is not serial, then Boolean negation cannot be expressed and some of the standard techniques used in completeness proofs for logics with common knowledge (and other fixpoint) operators cannot be used.

\appendix
\section{Proofs}

Let $L$ be any set of formulas containing all substitution instances of propositional tautologies in $\{ \land, \lor, \supset \}$ that is closed under Modus Ponens and Uniform substitution. We say that a set of formulas $\Delta$ is \emph{$L$-derivable} from a set of formulas $\Gamma$, notation $\Gamma \vdash_{L} \Delta$, iff there is $\gamma = \bigwedge \Gamma' \subseteq \Gamma$ and $\delta = \bigvee \Delta' \subseteq \Delta$ such that $\gamma \supset \delta$ is in $L$. We note that $\bigwedge \emptyset := \top$, so if $\Delta$ contains an element of $L$, then $\Gamma \vdash_{L} \Delta$ for all $\Gamma$. We say that $\langle \Gamma, \Delta\rangle$ is an \emph{independent $L$-pair} iff $\Gamma \not\vdash_{L} \Delta$.

A \emph{prime $L$-theory} is any set of formulas $\Gamma$ that i) contains $L$, ii) is closed under $\phi \supset \psi \in L$ (that is, if $\phi \supset \psi \in L$ and $\phi \in \Gamma$, then $\psi \in \Gamma$) and iii) contains $\phi \lor \psi$ only if it contains $\phi$ or $\psi$. A prime $L$-theory is called \emph{non-trivial} iff it is not the set of all formulas.

\begin{theorem}[Pair Extension]
If $\langle\Gamma,\Delta\rangle$ is an independent $L$-pair, then there is a non-trivial prime $L$-theory $\Sigma$ extending $\Gamma$ that is also disjoint from $\Delta$.
\end{theorem}
\begin{proof}
Essentially \cite[92--95]{Restall2000}. We note that $\bigvee \emptyset := \bot$, so $\Sigma$ cannot contain any $\chi$ such that $\chi \supset \bot$ is in $L$; hence $\Sigma$ has to be non-trivial.
\end{proof}

\noindent We note that in order for the Pair Extnesion Theorem to hold it is crucial to  define $L$-derivability in a ``finitary'' way; see \cite{Bilkova2018}.

\medskip

\noindent\textbf{Theorem \ref{thm: GrFDE complete decidable}.} \textit{$GrFDE$ is a sound and weakly complete axiomatization of the $\mathcal{L}_{Gr}$-theory of all group frames. The theory is decidable.}
\begin{proof}
Soundness is left to the reader as an exercise. Completeness is established using a variant of the standard finite canonical model construction (see e.g.\ \cite[Ch.\ 3.1]{Fagin1995}). The argument used here is based on \cite{Nishimura1982}.

Assume that $\phi_0$ is not provable in $GrFDE$. Let the closure of $\phi_0$, $Cl(\phi_0)$, be the smallest set of formulas that is closed under subformulas such that 1) it contains $\phi_0$; 2) it contains $\top$; 3) if $K^{*}_G \psi \in Cl(\phi_0)$, then $K_G (\psi \land K^{*}_G \psi) \in Cl(\phi_0)$; and 4) if $K_G \psi \in Cl(\phi_0)$, then $K_a \psi \in Cl(\phi_0)$ for all $a \in G$. Formula $\psi$ is a \emph{negated formula} iff $\psi$ is of the form $\Neg\chi$ for some formula $\chi$. We define $\tilde{\psi} := \Neg \psi$ in case $\psi$ is not a negated formula and $\tilde{\Neg\chi} := \chi$. Let $Cl^{\Neg}(\phi_0) = Cl(\phi_0) \cup \{ \tilde{\psi} \mid \psi \in Cl(\phi_0) \}$. It can be shown easily that $Cl^{\Neg}(\phi_0)$ is finite. We denote $Cl(\phi_0)$ as $\Phi$ and $Cl^{\Neg}(\phi_0)$ as $\Phi'$ in the rest of the proof.

We define a finite canonical model as follows. The set of situations $S$ is the set of all independent $GrFDE$-pairs $x = \langle x_{in}, x_{out}\rangle$ such that $x_{in} \cup x_{out} = \Phi'$. It can be shown that each independent $GrFDE$-pair $\langle \Gamma, \Delta\rangle$ can be extended to an independent $GrFDE$-pair $\langle \Gamma', \Delta' \rangle$ such that $\Phi' \subseteq (\Gamma' \cup \Delta')$. Note that, for all $x$, $x_{in}$ contains always at least $\top$. Otherwise $\top \in x_{out}$ and $x_{in} = \emptyset$, but then $\bigwedge x_{in} \supset \bigvee x_{out}$ is provable and so $x$ is not an independent pair. 

The rest of the model is defined as follows. The Routley star is defined by $x^{\star} := \langle x^{\star}_{in} = \{ \psi \in \Phi' \mid \tilde{\psi} \in x_{out} \}, (\Phi' \setminus x^{\star}_{in} ) \rangle$. It is easily seen that $x$ is an independent $GrFDE$-pair and thus an element of $S$ in the finite canonical model.
 Let us show that the canonical Routley star is of period two. It is clear that $\tilde{\tilde{\psi}} = \psi$ for all $\psi$. Therefore, $x^{\star\star}_{in} = \{ \psi \in \Phi' \mid \tilde{\psi} \in x^{\star}_{out} \}$ $= \{ \psi \in \Phi' \mid  \tilde{\tilde{\psi}} \notin x^{\star}_{out} \}$ $= x_{in}$.
 
Next, we define $R_axy$ iff $\{ \psi \mid K_a \psi \in x_{in} \} \subseteq y_{in}$. $R_a$ is reflexive thanks to (A6). The group relations $R_G$ and $R^{*}_G$ are defined as in ordinary group models. The canonical valuation is $v: p \mapsto \{ x \mid p \in x_{in}\}$ for $p \in \Phi$ and $v: p \mapsto \emptyset$ otherwise.

It remains to be shown that, for all $\psi \in \Phi$, $\psi \in x_{in}$ iff $x \vDash_v \psi$ (the Truth Lemma). For propositional variables, this holds by definition. It is easily seen that $\phi \land \psi \in x_{in}$ iff both $\phi,\psi \in x_{in}$ and $\phi \lor \psi \in x_{in}$ iff at least one of $\phi,\psi$ is in $x_{in}$, from which the claims for conjunctions and disjunctions follow. The claim for $\supset$ is similarly easy (it follows from from the fact that $FDE$ proves all positive classical tautologies---including $\top \supset (\phi \lor (\phi \supset \psi))$---and the fact that $\top \in x_{in}$ for all $x \in S$). 

The claims for the modal operators are established as follows. If $K_a \phi \in \Phi$, then $K_a \phi \in x_{in}$ implies $x \vDash_v K_a \phi$ by definition of $R_a$. To establish the converse implication, it is sufficient to observe that, if $K_a \phi \notin x_{in}$, then $\langle \{ \psi \mid K_a \psi \in x_{in} \}, \{ \phi \} \rangle$ is an independent pair. Hence, it can be extended to a pair $\langle \Gamma, \Delta\rangle$ such that $\Phi' \subseteq (\Gamma \cup \Delta)$. Take $y_{in} = \Gamma \cap \Phi'$ and $y_{out} = \Delta \cap \Phi'$. It is clear that $y = \langle y_{in}, y_{out}\rangle$ is an element of the canonical model such that $R_axy$ and that $y \not\vDash_v \phi$ (by the induction hypothesis).

The case of $K_G \phi \in \Phi$ where $G$ is not a singleton follows from (A7), the definition of $R_G$ and the induction hypothesis (for $K_a \phi$, $a \in G$; note that we may use the hypothesis as $K_G \phi \in \Phi$ implies $K_a \phi \in \Phi$ for all $a \in G$).

Finally, take $K^{*}_G \phi \in \Phi$. If $K^{*}_G \phi \in x_{in}$, then $x \vDash_v K^{*}_G \phi$ by the fact that $K_G (\phi \land K^{*}_G \phi) \in \Phi$ and the induction hypothesis for $K_G$. The converse entailment is established as follows. For each non-empty $Z \subseteq S$ and $y \in S$ of the finite canonical model, define
\[
\phi_y := \bigwedge y_{in} \qquad
\phi_Z := \bigvee_{y \in Z} \phi_y
\]  
We sometimes write $y$ instead of $\phi_y$ and $Z$ instead of $\phi_Z$. Take $Z := \{ y \mid R^{*}_G xy \}$ and assume that $\phi \in y_{in}$ for all $y \in Z$. We have to prove that $K^{*}_{G}\phi \in x_{in}$.
\begin{lemma}\label{lemma: common-1}
$GrFDE$ proves $Z \supset K_G (\phi \land Z)$.
\end{lemma} 
\noindent Before proving the lemma, we show how it is applied. Using the Induction Rule (R3) and the fact that $x \in Z$, we obtain $\vdash x \supset K^{*}_G \phi$. Hence, $K^{*}_G \phi$ must be in $x_{in}$, otherwise $x$ would not be an independent pair.

\noindent\textit{Proof of Lemma \ref{lemma: common-1}.} We write $X_GY$ if, for all $x \in X$, if $R_G xy$, then $y \in Y$. We prove the following claim.

\begin{lemma}\label{lemma: common-2}
If $X_GY$, then $GrFDE$ proves $X \supset K_G Y$.
\end{lemma}

\noindent\textit{Proof of Lemma \ref{lemma: common-2}.} We prove that if $X_GY$, then $GrFDE$ proves $X \supset K_a Y$ for all $a \in G$; the desired result then follows by applying axiom (A7). The proof is by reductio ad absurdum. Assume that $GrFDE$ does not prove $X \supset K_a Y$. Then there is $w \in X$ such that $GrFDE$ does not prove $w \supset K_a Y$. This means that $\langle \{ \psi \mid K_a \psi \in w_{in} \}, \{ \phi_z \mid z \in Y\} \rangle$ is an independent pair extendible to $\langle \Gamma, \Delta\rangle$ such that $\Gamma \cup \Delta$ contains $\Phi'$. Take $y = \langle \Gamma \cap \Phi', \Delta \cap \Phi' \rangle$. It is clear that $R_a wy$ and so, by our assumption, $y \in Y$. However, this means that $GrFDE$ proves $\phi_y \supset \phi_Y$ and so $y$ cannot be an independent pair. This is a contradiction. Hence, Lemma \ref{lemma: common-2} is established.

We continue the proof of Lemma \ref{lemma: common-1}. Note that $Z_GZ$, so $GrFDE$ proves $Z \supset K_G Z$ by Lemma \ref{lemma: common-2}. Moreover, our assumption that $\phi \in y_{in}$ for all $y \in Z$ implies that $GrFDE$ proves $Z \supset \phi$. Using monotonicity and regularity of $K_G$, we infer that $GrFDE$ proves $K_G Z \supset K_G( \phi \land Z)$. Hence, $GrFDE$ proves $Z \supset K_G (\phi \land Z)$ as desired. This concludes the proof of Lemma \ref{lemma: common-1} and completeness is established.

Our proof shows that the $\mathcal{L}_{Gr}$-theory of group frames is recursively axiomatizable and the axiomatization is complete with respect to a recursively enumerable set of models (models based in finite group frames). Hence, the theory is decidable.
\end{proof}

\medskip

\noindent\textbf{Theorem \ref{thm: GrUpFDE complete decidable}.} \textit{$GrUpFDE$ is a sound and weakly complete axiomatization of the $\mathcal{L}_{GrUp}$-theory of all group update frames. The theory is decidable.}
\begin{proof}
Assume that $\phi_0$ is not provable in $GrUpFDE$. Define the finite canonical model based on the closure of $\phi_0$, $\Phi$, and the $\Neg$-closure of $\Phi$, which we denote $\Phi'$, as in the proof of Theorem \ref{thm: GrFDE complete decidable}. Moreover, let $RxYz$ iff there are prime $GrUpFDE$-theories $\Gamma, \Sigma$ and $\Delta_i$ for $i \in I$ such that
\begin{itemize}
\item[(a)] for all $\phi\LD\psi$, if $\phi\LD\psi \in \Gamma$ and $\phi \in \bigcap_{i \in I} \Delta_i$, then $\psi \in \Sigma$; and
\item[(b)] $x_{in} \subseteq \Gamma$, $(\Sigma \cap \Phi') \subseteq z_{in}$ and, for all $\Delta_i$ there is $y_j \in Y$ such that $(y_j)_{in} \subseteq \Delta_i$.
\end{itemize} 
(A similar definition appears in \cite{Bilkova2016}.) We have to show only that the Truth Lemma holds for $\phi\LD\psi \in \Phi$. If $\phi\LD\psi \in x_{in}$, $RxYz$ and $\phi \in \bigcup \{y_{in} \mid y \in Y \}$, then $\psi \in z_{in}$ by the definition of the canonical $R$. Conversely, we reason similarly as in \cite[256]{Restall2000}.  First, assume that $\phi\LD\psi \in x_{out}$. Extend $x$ to a prime theory $\Gamma$. Second, extend the independent pair $\langle \{ \chi \mid \phi\LD\chi \in \Gamma \}, \{ \psi \}\rangle$ to a prime theory $\Sigma$. (The proof that it is an independent pair uses (A9) and (R4); the case $\{ \chi \mid \phi\LD\chi \in \Gamma \} = \emptyset$ uses (R5).) Third, take the set $\Lambda = \{ \alpha \mid \exists \beta (\beta \notin \Sigma \And \alpha \LD \beta \in \Gamma \}$. For each $\alpha_i \in \Lambda$, $\phi \supset \alpha_i$ is not provable. (If some $\phi \supset \alpha_i$ were provable, then $\beta_i \in \Sigma$ by (R4).) Hence, extend each pair $\langle \{ \phi \}, \{ \alpha_i \}\rangle$ to a prime theory $\Delta_i$. It follows from the construction of $\Sigma$ and $\Delta_i$ that (a) holds for $\Gamma, \{ \Delta_i \}_{i \in I}, \Sigma$. (If $\Lambda = \emptyset$, then $\{ \Delta_i \}_{i \in I} = \emptyset$ and so each $\alpha \in \bigcap_{i \in I} \Delta_i$; but in this case also $\alpha \LD \beta \in \Gamma$ implies $\beta \in \Sigma$.) Moreover, $\phi \in \bigcap_{i \in I} \Delta_{i}$ by the construction of $\Delta_i$. Finally, take $(y_i)_{in} = \Delta_{i} \cap \Phi'$ and $(y_i)_{out}$ its complement relative to $\Phi'$ and similarly for $z$ and $\Sigma$. It is clear that $\phi \in (y_i)_{in}$ for all $i \in I$ and $\psi \notin z_{in}$. This concludes the proof of the Truth Lemma.

Our proof establishes that the $\mathcal{L}_{GrUp}$-theory of group update frames is recursively axiomatizable and the axiomatization is complete with respect to a recursively enumerable set of models (models based on finite group update frames). Hence, the theory is decidable.
\end{proof}



\begin{thebibliography}{10}
\providecommand{\url}[1]{\texttt{#1}}
\providecommand{\urlprefix}{URL }
\providecommand{\doi}[1]{https://doi.org/#1}

\bibitem{Aucher2016}
Aucher, G.: {Dynamic epistemic logic in update logic}. Journal of Logic and
  Computation  \textbf{26}(6),  1913--1960 (03 2016).
  \doi{10.1093/logcom/exw002}, \url{https://doi.org/10.1093/logcom/exw002}

\bibitem{Barwise1983}
Barwise, J., Perry, J.: {Situations and Attitudes}. MIT Press (1983)

\bibitem{Belnap1977a}
Belnap, N.: {A useful four-valued logic}. In: Dunn, J.M., Epstein, G. (eds.)
  Modern Uses of Multiple-Valued Logic, pp. 5--37. Springer Netherlands,
  Dordrecht (1977)

\bibitem{Belnap1977b}
Belnap, N.: {How a computer should think}. In: Ryle, G. (ed.) Contemporary
  Aspects of Philosophy. Oriel Press Ltd. (1977)

\bibitem{Benthem1991}
van Benthem, J.: {Language in Action: Categories, Lambdas and Dynamic Logic}.
  Elsevier Science Publishers, Amsterdam (1991)

\bibitem{Benthem2011}
van Benthem, J.: {Logical Dynamics of Information and Interaction}. Cambridge
  University Press (2011)

\bibitem{Benthem2008}
van Benthem, J.: Logical dynamics meets logical pluralism? The Australasian
  Journal of Logic  \textbf{6} (2008). \doi{10.26686/ajl.v6i0.1801},
  \url{https://ojs.victoria.ac.nz/ajl/article/view/1801}

\bibitem{Benthem2014}
van Benthem, J.: Two logical faces of belief revision. In: Trypuz, R. (ed.)
  {Krister Segerberg on Logic of Actions}, pp. 281--300. Springer Netherlands,
  Dordrecht (2014)

\bibitem{Bilkova2018}
B{\'i}lkov{\'a}, M., Cintula, P., L{\'a}vi{\v{c}}ka, T.: Lindenbaum and pair
  extension lemma in infinitary logics. In: Moss, L.S., de~Queiroz, R.,
  Martinez, M. (eds.) {Logic, Language, Information, and Computation
  (Proceedings of WoLLIC 2018)}. pp. 130--144. Springer Berlin Heidelberg,
  Berlin, Heidelberg (2018)

\bibitem{Bilkova2016}
B\'ilkov\'a, M., Majer, O., Peli\v{s}, M.: Epistemic logics for sceptical
  agents. Journal of Logic and Computation  \textbf{26}(6),  1815--1841 (2016)

\bibitem{Bimbo2005}
Bimb{\'o}, K., Dunn, J.M.: Relational semantics for {Kleene Logic} and {Action
  Logic}. Notre Dame J. Formal Logic  \textbf{46}(4),  461--490 (10 2005).
  \doi{10.1305/ndjfl/1134397663},
  \url{https://doi.org/10.1305/ndjfl/1134397663}

\bibitem{Chellas1975}
Chellas, B.F.: Basic conditional logic. Journal of Philosophical Logic
  \textbf{4}(2),  133--153 (May 1975). \doi{10.1007/BF00693270},
  \url{https://doi.org/10.1007/BF00693270}

\bibitem{vanDitmarsch2008}
van Ditmarsch, H., van~der Hoek, W., Kooi, B.: {Dynamic Epistemic Logic}.
  Springer (2008)

\bibitem{Dosen1992}
Do\v{s}en, K.: A brief survey of frames for the {Lambek Calculus}. Mathematical
  Logic Quarterly  \textbf{38}(1),  179--187 (1992).
  \doi{10.1002/malq.19920380113},
  \url{http://dx.doi.org/10.1002/malq.19920380113}

\bibitem{Dunn2001a}
Dunn, J.M.: Ternary relational semantic and beyond: {P}rograms as data and
  programs as instructions. Logical Studies  \textbf{7} (2001)

\bibitem{Fagin1995}
Fagin, R., Halpern, J.Y., Moses, Y., Vardi, M.Y.: {Reasoning About Knowledge}.
  MIT Press (1995)

\bibitem{Fagin1995a}
Fagin, R., Halpern, J.Y., Vardi, M.: {A nonstandard approach to the logical
  omniscience problem}. Artificial Intelligence  \textbf{79},  203--240 (1995).
  \doi{10.1016/0004-3702(94)00060-3}

\bibitem{Girard2016}
Girard, P., Tanaka, K.: Paraconsistent dynamics. Synthese  \textbf{193}(1),
  1--14 (Jan 2016). \doi{10.1007/s11229-015-0740-2},
  \url{https://doi.org/10.1007/s11229-015-0740-2}

\bibitem{Holliday2012}
Holliday, W.H., Hoshi, T., Icard~III, T.F.: A uniform logic of information
  dynamics. In: Bolander, T., Bra\"{u}ner, T., Ghilardi, S., Moss, L. (eds.)
  {Advances in Modal Logic 2012}, pp. 348--367. College Publications (2012)

\bibitem{Kurtonina1994}
Kurtonina, N.: {Frames and Labels. A Modal Analysis of Categorial Inference}.
  {PhD Thesis}, Utrecht University (1994)

\bibitem{Kurucz1995}
Kurucz, {\'A}., N{\'e}meti, I., Sain, I., Simon, A.: Decidable and undecidable
  logics with a binary modality. Journal of Logic, Language and Information
  \textbf{4}(3),  191--206 (Sep 1995). \doi{10.1007/BF01049412},
  \url{https://doi.org/10.1007/BF01049412}

\bibitem{Levesque1984}
Levesque, H.: {A logic of implicit and explicit belief}. In: {Proceedings of
  AAAI 1984}. pp. 198--202 (1984)

\bibitem{Mares2002}
Mares, E.D.: A paraconsistent theory of belief revision. Erkenntnis
  \textbf{56}(2),  229--246 (Mar 2002). \doi{10.1023/A:1015690931863},
  \url{https://doi.org/10.1023/A:1015690931863}

\bibitem{Mares2004}
Mares, E.D.: {Relevant Logic: A Philosophical Interpretation}. Cambridge
  University Press, Cambridge (2004)

\bibitem{Miller2005}
Miller, J.S., Moss, L.S.: The undecidability of iterated modal relativization.
  Studia Logica  \textbf{79}(3),  373--407 (Apr 2005).
  \doi{10.1007/s11225-005-3612-9},
  \url{https://doi.org/10.1007/s11225-005-3612-9}

\bibitem{Nishimura1982}
Nishimura, H.: Semantical analysis of constructive {PDL}. Publications of the
  Research Institute for Mathematical Sciences  \textbf{18}(2),  847--858
  (1982). \doi{10.2977/prims/1195183579}

\bibitem{Restall1995a}
Restall, G.: Information flow and relevant logic. In: Seligman, J.,
  Westersh\.{a}hl, D. (eds.) {Logic, Language and Computation: The 1994 Moraga
  Proceedings}, pp. 463--477. CSLI Press (1995)

\bibitem{Restall1995}
Restall, G., Slaney, J.: Realistic belief revision. In: {Proceedings of the
  Second World Conference on Foundations of Artificial Intelligence}. pp.
  367--378 (1995)

\bibitem{Restall2000}
Restall, G.: {An Introduction to Substrucutral Logics}. Routledge, London
  (2000)

\bibitem{Rivieccio2014}
Rivieccio, U.: {Bilattice Public Announcement Logic}. In: Gor\'{e}, R., Kooi,
  B., Kurucz, A. (eds.) {Advances in Modal Logic 2014}. pp. 459--477. College
  Publications (2014)

\bibitem{Sedlar2016}
Sedl\'{a}r, I.: Propositional dynamic logic with {B}elnapian truth values. In:
  {Advances in Modal Logic. Vol. 11}. College Publications, London (2016)

\bibitem{Sedlar2019a}
Sedl\'ar, I., Pun\v{c}och\'{a}\v{r}, V.: From positive {PDL} to its
  non-classical extensions. Logic Journal of the IGPL \textbf{27}(4), 522--542  (2019).
  \doi{10.1093/jigpal/jzz017}.

\end{thebibliography}

\end{document}